\documentclass[amsmath,amssymb,aps,pra,reprint,final]{revtex4-1}

\usepackage{graphicx}
\usepackage{amsthm}
\usepackage{bm,braket,mathtools}
\usepackage[unicode]{hyperref}

\newtheorem{lemma}{Lemma}
\newtheorem{proposition}{Proposition}
\newtheorem{theorem}{Theorem}
\theoremstyle{remark}
\newtheorem{remark}{Remark}

\def\Affiliation{QUIT group, Dipartimento di Fisica, Università di Pavia, and INFN Sezione di Pavia, via Bassi 6, 27100 Pavia, Italy}

\newcommand\spc[1]{\mathcal{#1}}
\def\sH{\spc H}
\def\Psucc{\ensuremath{\mathcal{P}_{\text{s}}}}
\def\SEPSet{\ensuremath{\mathsf{SEP}}}
\def\LOCCSet{\ensuremath{\mathsf{LOCC}}}
\def\FockSpace{\ensuremath{\mathcal{F}}}

\newcommand{\X}[1]{\abs{\braket{\tilde{\psi}_{#1} | \tilde{\phi}_{#1}}}}
\newcommand{\Y}[1]{\abs{\braket{{\psi}_{#1} | {\phi}_{#1}}}}
\DeclareMathOperator{\St}{St}
\DeclareMathOperator{\Tr}{Tr}
\DeclareMathOperator{\Supp}{Supp}

\DeclarePairedDelimiter\abs{\lvert}{\rvert}
\DeclarePairedDelimiter\norm{\lVert}{\rVert}

\usepackage[smaller,nolist]{acronym}

\newcommand{\orcid}[1]{\href{#1}{\includegraphics[height=.8em]{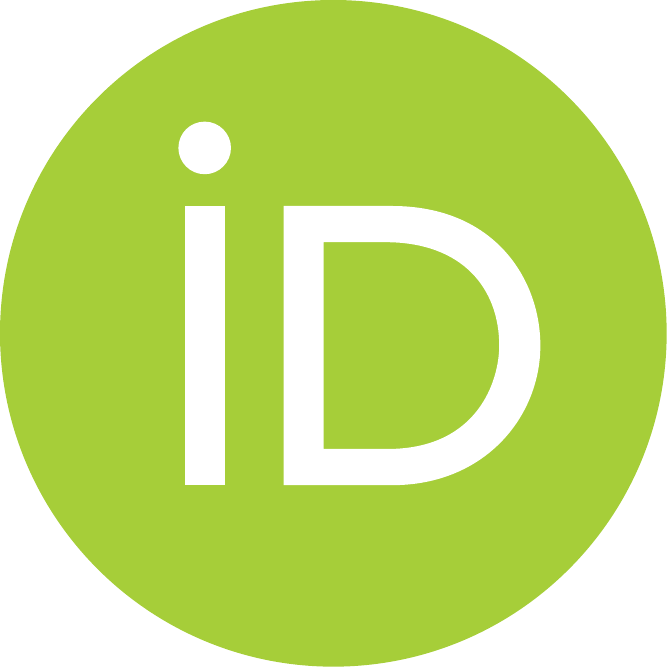}}}

\begin{document}
	\title{Unambiguous discrimination of Fermionic states through local operations and classical communication}
	
	\author{Matteo \surname{Lugli} \orcid{https://orcid.org/0000-0003-0554-3760}}
	\email{matteo.lugli01@ateneopv.it}
	\affiliation{\Affiliation}
	
	\author{Paolo \surname{Perinotti} \orcid{https://orcid.org/0000-0003-4825-4264}}
	\email{paolo.perinotti@unipv.it}
	\affiliation{\Affiliation}
	
	\author{Alessandro \surname{Tosini} \orcid{https://orcid.org/0000-0001-8599-4427}}
	\email{alessandro.tosini@unipv.it}
	\affiliation{\Affiliation}
	
	\begin{abstract}
		\acrodef{LOCC}{local operations and classical communication}
		\acrodef{POVM}{positive-operator valued measure}
		\acrodef{SEP}{separable effects}\acused{SEP}
		The paper studies unambiguous discrimination of Fermionic states
		through \ac{LOCC}.
		In the task of unambiguous discrimination, no error
		is tolerated but an inconclusive result is allowed.
		We show that contrary to the quantum case, it is not always possible to
		distinguish two Fermionic states through \ac{LOCC} unambiguously
		with the same success probability as if global measurements were
		allowed. Furthermore, we prove that we can overcome such a limit
		through an ancillary system made of two Fermionic modes,
		independently of the dimension of the system, prepared in a
		maximally entangled state: in this case, \ac{LOCC} protocols
		achieve the optimal success probability.
	\end{abstract}
	
	\maketitle
	
	\section{Introduction}	
	The quest for state discrimination has been thoroughly investigated in the quantum realm \cite{Helstrom1969,Helstrom1976,Ivanovic1987,Dieks1988,Peres1988,Chefles2000}.
	When dealing with composite systems, the peculiar nature of entanglement, which quantum systems can exhibit, gives rise to counterintuitive situations where the information is encoded
	in a delocalized fashion, contradicting the common wisdom according to which any 
	information carried by a system should be encoded in its own local degrees of freedom. 
	As a consequence, in order to discriminate a full basis of entangled pure states, a delocalized measurement is required.
	However, when processing distributed information, a reliable channel for the exchange of quantum systems is an expensive resource and one is thus willing to accept lower performances if it is possible to achieve a reasonable result using only \acf{LOCC} (for a 
	comprehensive reference on \ac{LOCC} protocols, see Ref.~\cite{Chitambar2014}).
	
	For this reason, discrimination of two pure states via \ac{LOCC} is twice  counterintuitive: indeed, classic results in quantum information theory proved that \ac{LOCC} discrimination of pure states can be as good as the optimal one, both in a minimum-error scenario~\cite{Walgate2000,Virmani2001} and in a zero-error scenario, where one accepts an inconclusive outcome~\cite{Ji2005}.
	
	In the present work, we study the problem of optimal unambiguous discrimination of pure states via \ac{LOCC} in Fermionic theory.
	This theory describes systems, states, measurements, and transformations on Fermions, or, more precisely, on Fermionic modes.
	The interest in Fermionic theory is clearly due to the fact that elementary matter fields in quantum physics are actually collections of Fermionic modes.
	While it is well known that from a computational
	point of view, Fermionic modes are equivalent to
	qubits~\cite{Bravyi2002}, to the extent that quantum algorithms have
	been devised for the simulation of scattering processes involving
	Fermions~\cite{Jordan2014}, still the differences between the two
	theories are relevant from the point of view of the complexity of
	geometric structures required for such simulation: indeed, there is no
	local encoding of Fermionic modes in qubit networks, or
	vice versa~\cite{Wolf2006,Banuls2007,DAriano2014a}.
	As a matter of fact, then, while the computational equivalence of qubits and Fermions holds asymptotically, modulo a polynomial overhead of program complexity, for finite-size problems it is relevant to study simple information-processing tasks with Fermionic modes autonomously. Among many options for the description of Fermionic modes through a suitable
	choice of algebraic structures~\cite{Friis2013,DAriano2014}, it is
	often convenient to resort to the use of the Jordan-Wigner
	isomorphism~\cite{Jordan1928,Verstraete2005,Pineda2010}, which allows
	one to map the algebra of $N$ canonically anticommuting field
	operators into strings of Pauli operators acting on the Hilbert space
	of $N$ qubits, with a parity superselection
	rule~\cite{Kitaev2004,Schuch2004,Schuch2004a}.
	
	When dealing with state discrimination, where all the relevant formulas come from the application of the Born rule for the calculation of probabilities, the Jordan-Wigner representation is particularly convenient, and the only care that must be taken in
	representing Fermions with qubits is an appropriate account for parity
	superselection~\cite{DAriano2014}, as shown in
	Ref.~\cite{Lugli2020} where the task of optimal minimum-error
	discrimination was studied.
	
	Here we take a similar approach to the problem of optimal unambiguous state 
	discrimination. One of the key features of \ac{LOCC} discrimination is that unlike the case of binary quantum state discrimination, ancilla-assisted 
	protocols~\cite{Jonathan1999} prove to be useful. More precisely, in the minimum-error
	scenario, one entangled pair is sufficient to make optimal \ac{LOCC} discrimination 
	equivalent to its unconstrained counterpart.
	
	The present work extends that of Ref.~\cite{Lugli2020}, which
	deals with both perfect and optimal conclusive discrimination of
	Fermionic states.  In order to distinguish between two quantum
	non orthogonal preparations, the further strategy of unambiguous
	discrimination has been proposed by~\citeauthor{Jaeger1995}
	in Ref.~\cite{Jaeger1995}: we require the outcomes to flawlessly
	determine the true state.
	However, due to the essential nature of
	quantum theory, the whole protocol must make allowance for a third
	inconclusive result, for which no intelligence is attained from the
	measurement. In Refs.~\cite{Chen2001,Chen2002,Ji2005}, the authors
	prove that restricting ourselves to \ac{LOCC} protocols is no real
	limit to the performances of unambiguous discrimination.  We briefly
	discuss the protocol in the Fermionic realm and prove that Fermionic
	local discrimination is typically \textsl{suboptimal} with respect to
	the unconstrained one, contrary to the quantum case.  The necessary
	and sufficient conditions on the states are derived for achieving
	optimal local discrimination.  Moreover, we prove that \ac{LOCC}
	protocols accomplish the same discrimination performances as the
	unconstrained ones, once we provide an ancillary system in a maximally
	entangled state.

	\section{Preliminary notions}
	
	\subsection{Unambiguous quantum states discrimination}\label{sec:quanutm-disc}
	We are interested in discriminating between two non-normalized and
	pure states $\rho = p\ket\psi\bra\psi$ and
	$\sigma = q\ket\phi\bra\phi$, where $\{p, q\}$ is the prior
	probability distribution and $\ket\psi$, $\ket\phi$ are normalized
	vectors.
	In quantum theory, we know that if the two states are orthogonal, we may perfectly distinguish between the two.
	If we release such a condition, on the contrary, the discrimination protocol becomes probabilistic and subject to errors.
	We remind the reader that a quantum measurement is generally represented by a \ac{POVM}, namely, a collection of positive operators $0 \le S \le I$---called \emph{effects}---that sums to the identity operator $I$.
	Every operator $\Pi_i$ in a \ac{POVM} is
	associated with a possible outcome $i$, and the probability of outcome
	$i$, provided that the measurement is performed on a system prepared
	in state $\rho$, is given by the Born rule
	\begin{align}
		\Pr(i|\rho)=\Tr[\rho\Pi_i].
	\end{align}
	
	If there exist operators $0 \le A_i, B_i \le I$ such that
	$S = \sum_i A_i \otimes B_i$, we call the effect $S$ \emph{separable}.
	If a \ac{POVM} is made of separable effects, then we call it separable as well, and denote with \ac{SEP} the set of separable \acp{POVM}.
	Moreover, we recall that \ac{LOCC} measurements are a proper subset of \ac{SEP} \acp{POVM}~\cite{Bennett1999} (for further details on \ac{LOCC} protocols, see Ref.~\cite{Chitambar2014}).
	
	In the following, we deal with discrimination strategies that unambiguously distinguish between the two provided states $\rho$ and $\sigma$, though allowing for an inconclusive result.
	We describe the measurement through the 
	\ac{POVM} made of $\Pi_\psi$, $\Pi_\phi$, and $\Pi_?$, where each operator corresponds to one of the three possible outcomes.
	The requirement for a \ac{POVM} to represent an unambiguous discrimination protocol is given by the following conditions:
	\begin{align}
		\label{eq:unambiguous_condition}
		\left\{
		\begin{aligned}
			&\Pr(\phi|\psi)=\Tr[\ket\psi\bra\psi \Pi_\phi] =0, \\
			&\Pr(\psi|\phi)=\Tr[\ket\phi\bra\phi \Pi_\psi] = 0,
		\end{aligned}
		\right.
	\end{align}	
	and $0\leq\Pi_\psi+\Pi_\phi\leq I$. Under these circumstances, we define $\Pi_? \coloneqq I - \Pi_\psi - \Pi_\phi$.
	The success probability of the protocol is then given by
	\begin{equation*}
		\Psucc \coloneqq \Tr[(p \ket\psi\bra\psi + q \ket\phi\bra\phi) (\Pi_\psi + \Pi_\phi)]
	\end{equation*}
	or $\Psucc = 1 - \mathcal{P}_{\text{err}}$, where $\mathcal{P}_{\text{err}} \coloneqq \Tr[(p \ket\psi\bra\psi + q \ket\phi\bra\phi) \Pi_?]$ is the error probability.
	
	In Ref.~\cite{Jaeger1995}, the authors describe the optimal \ac{POVM} that maximizes the probability of discrimination success $\Psucc$ for the provided states $\rho$, $\sigma$.
	Due to the effects being dominated by the identity, the set of optimal \acp{POVM} splits into two classes depending on the relationship between the scalar product $\braket{\psi|\phi}$ of the two preparations and the relevant quantity of
	\begin{equation}\label{eq:Xi}
		\Xi(p, q) \coloneqq \sqrt{\frac{\min\{p, q\}}{\max\{p, q\}}} , \quad p, q > 0 .
	\end{equation} 
	Indeed, the optimal \ac{POVM} achieves
	\begin{align}
		\Psucc (\rho, \sigma) &= 1 - 2\sqrt{p q} \abs{\braket{\psi|\phi}} \label{eq:psucc_ternary} \\
		\intertext{for $\abs{\braket{\psi|\phi}} \le \Xi(p, q)$, whereas}
		\Psucc (\rho, \sigma) &= \max\{p, q\} \left(1 - \abs{\braket{\psi|\phi}}^2\right) \label{eq:psucc_binary}
	\end{align}
	otherwise.
	In the latter case, the probabilities are so unbalanced that it is convenient to renounce detection of the least probable state, by letting its \ac{POVM} operator be null.
	Henceforward, we will denote such a protocol as \textsl{binary} since the measurement consists of only two
	effects. Consistently, the former optimal strategy will be
	deemed \textsl{ternary}.  One may prove with ease that the two
	success probabilities of Eqs.~\eqref{eq:psucc_ternary} and
	\eqref{eq:psucc_binary} satisfy the following inequality:
	\begin{equation}\label{eq:binary_versus_ternary}
		\max\{p, q\} \left(1 - \abs{\braket{\psi|\phi}}^2\right) \le 1 - 2\sqrt{p q} \abs{\braket{\psi|\phi}} ,
	\end{equation}
	where equality is achieved iff $\abs{\braket{\psi|\phi}} = \Xi(p, q)$.
	
	Here we focus on states representing preparations of a bipartite system $\mathrm{AB}$ shared between Alice and Bob.
	In Ref.~\cite{Ji2005}, the authors show that the optimal success probability, given by  
	Eq.~\eqref{eq:psucc_ternary} or \eqref{eq:psucc_binary}, depending on the circumstances, can be achieved in the bipartite scenario by an \ac{LOCC} protocol.
	Namely, Alice has to carry out a measurement on her party in a suitably chosen orthonormal basis, and then send the outcome through a classical channel to Bob.
	At this stage, Bob either perfectly or unambiguously discriminates between two local states and estimates the correct result.
	
	\subsection{Fermionic quantum theory}
	The Fermionic quantum theory describes states and transformations of local Fermionic modes (see Refs.~\cite{Bravyi2002,Wolf2006,Banuls2007,Friis2013,DAriano2014}) that satisfy the parity \emph{superselection rule}~\cite{Wick1952,Hegerfeldt1968,Schuch2004,Kitaev2004,Schuch2004a,DAriano2014a,Bravyi2002}.
	A Fermionic
	mode can be either empty or ``excited'' and vectors
	representing Fermionic systems are allowed to be superimposed
	only if they exhibit the same parity, namely, the excitation
	numbers are all either even or odd.  Such a rule is equivalent
	to requiring that local Fermionic transformations are described
	by Kraus operators belonging to the Fermionic algebra
	$\mathfrak{F}$.  The generators of $\mathfrak{F}$ are the
	operators $\varphi_i$, for $i$ from 1 to $N$ number of
	modes, fulfilling the canonical anticommutation relations
	$\{\varphi_i, \varphi_j^\dagger\} = \delta_{ij}$ and
	$\{\varphi_i, \varphi_j\} =
	\{\varphi_i^\dagger,\varphi_j^\dagger\} = 0,\ \forall i, j$.
	The Fermionic operators enable us to construct the Fock states
	as
	$\ket{n_1\ldots n_N} \coloneqq (\varphi_1^\dagger)^{n_1}
	\cdots (\varphi_N^\dagger)^{n_N} \ket{\Omega}$,
	where the vacuum state $\ket{\Omega}$ is the common
	eigenvector of operators $\varphi^\dag_i\varphi_i$ with null
	eigenvalues, and $n_i$ corresponds to the occupation
	number at the $i$th mode, i.e., the expectation value of the
	operator $\varphi^\dag_i\varphi_i$. The linear span of all
	Fock states corresponds to the antisymmetric Fock space
	$\FockSpace$.  We may denote by $\FockSpace_e$ and
	$\FockSpace_o$ those sectors of the Fock space featuring even
	and odd parity, respectively, with
	$\FockSpace=\FockSpace_e\oplus\FockSpace_o$. Every state (or
	effect) has a well-defined parity, i.e., states (and effects)
	satisfy the parity superselection rule.
	
	Thanks to the Jordan-Wigner isomorphism~\cite{Jordan1928,Verstraete2005,Pineda2010}, the Fock space
	$\FockSpace$ of $N$ Fermionic modes is isomorphic to a
	$N$-qubit Hilbert space, by the trivial identification of the
	Fermionic occupation number basis $\ket{n_1\ldots n_N}$ with
	the qubit computational basis (eigenvectors of the Pauli
	matrices $\sigma_z$ with $1\leq i \leq N$).  The two parities of
	the set of states (and effects) are actually isomorphic to the
	($N - 1$)-qubit states set, with even and odd pure states
	given by the rank-one projectors $\ket{\psi}\bra{\psi}$, with
	the $\ket{\psi}$ normalized superposition of Fock vectors
	belonging to $\FockSpace_e$ and $\FockSpace_o$, respectively.
	
	Hereafter, we take advantage of the Jordan-Wigner isomorphism to handle
	transformations and informational protocols in the Fermionic
	theory.  In fact, the isomorphism, which maps nonlocally, the
	Fermionic operator algebra to an algebra of transformations on
	qubits, allows us to proceed with the usual quantum
	notation. In this paper, we will focus on Fermionic state
	discrimination via \ac{LOCC}; it is therefore fundamental
	to characterize \ac{LOCC} for Fermionic systems. In this
	respect, the Jordan-Wigner transformation plays a major
	role. Indeed, in Ref.~\cite{DAriano2014}, the authors show that
	every Fermionic \ac{LOCC} corresponds to a quantum
	\ac{LOCC} on qubits by means of the Jordan-Wigner transformation $\mathcal{J}$.
	Accordingly, in the study of \ac{LOCC} Fermionic state discrimination, we will use the fact that in a fixed parity sector (even or odd), the results of quantum theory immediately hold in the Fermionic case also.
	
	\section{Fermionic states discrimination}
	Let us consider two states representing preparations of a Fermionic system.
	The theoretical result of Sec.~\ref{sec:quanutm-disc}, i.e., optimal quantum unambiguous discrimination, can be applied in the Fermionic case as well and this allows us to plainly reuse Eqs.~\eqref{eq:psucc_ternary} and \eqref{eq:psucc_binary} in the Fermionic realm too.
	
	Henceforth, we focus on the optimal discrimination strategies to distinguish two states $\rho = p\ket\psi\bra\psi$ and $\sigma = q\ket\phi\bra\phi$ of a Fermionic bipartite system $\mathrm{AB}$ via \ac{LOCC} protocols.
	First, we point out that if the two vectors $\ket\psi$, $\ket\phi$ feature a different parity, e.g., $\ket\psi \in \FockSpace_e(\mathrm{AB})$ and $\ket\phi \in \FockSpace_o(\mathrm{AB})$, they are perfectly discriminable, as 
	shown in Ref.~\cite{Lugli2020}.
	Hence, we are interested in discriminating states belonging to the same Fock space sector.
	Second, since the even and odd sectors are equivalent under \ac{LOCC}, it is not restrictive to focus on even vectors only.
	
	In the following, when dealing with a composite system of $N$ Fermionic modes made of two subsystems of $N_1$ and $N_2$ modes, respectively, with $N_1+N_2=N$, it is useful to introduce the spaces
	$\sH_E\coloneqq\sH_{e1}\otimes\sH_{e2}$ and $\sH_O\coloneqq\sH_{o1}\otimes\sH_{o2}$, so that $\sH_e=\sH_E\oplus\sH_O$. In other words, $\sH_E$ and $\sH_O$ are the subspaces where the parities of Alice's 
	and Bob's subsystems are both even or odd, respectively.
	Operators $X$ with both support and range in $\sH_E$ will be denoted by $X_E$, and similarly 
	operators $X$ with both support and range in $\sH_O$ will be denoted by $X_O$. In particular, we will often use the projections $P_E$ and $P_O$ on $\sH_E$ and $\sH_O$, respectively.
	
	Let us consider, for instance, the vector $\ket\psi$ and introduce the bases $\{\ket{e_i}_A\}$ and $\{\ket{o_j}_A\}$ for Alice, where $\ket{e_i}_A \in \FockSpace_e (\mathrm{A})$ and 
	$\ket{o_j}_A \in \FockSpace_o (\mathrm{A})$. Thanks to the superselection rule and the Schmidt decomposition, we can write the vector as
	\begin{equation*}
		\ket\psi = \ket{\psi_E} + \ket{\psi_O} ,
	\end{equation*}
	where $\ket{\psi_E} = \sum_i \ket{e_i}_A \otimes \ket{e_i'}_B$ and $\ket{\psi_O} = \sum_j \ket{o_j}_A \otimes \ket{o_j'}_B$ for some $\ket{e_i'}_B \in \FockSpace_e (\mathrm{B})$ and $\ket{o_j'}_B \in \FockSpace_o (\mathrm{B})$.
	Due to the above assumption, the scalar product between the two states $\ket\psi$, $\ket\phi$ reads
	\begin{equation}\label{eq:scalar_product}
		\braket{\psi | \phi} = \braket{\psi_E | \phi_E} + \braket{\psi_O | \phi_O} .
	\end{equation}
	
	Before dealing with Fermionic discrimination protocols, we first discuss some relevant properties of \ac{SEP} effects in the Fermionic theory.
	In order to satisfy the parity superselection rule, any \ac{SEP} \ac{POVM} must be made of positive operators $0 \le S \le I$ of the form
	\begin{equation}\label{eq:SEP_effect}
		S = S_E + S_O ,
	\end{equation}
	where $S_E = \sum_i e_i \otimes e_i'$,
	$S_O = \sum_j o_j \otimes o_j'$ for
	$0 \le e_i, e_i', o_j, o_j' \le I$. More precisely, the latter
	operators fulfill
	$\Supp(e_i) \subseteq \FockSpace_e (\mathrm{A})$,
	$\Supp(e_i') \subseteq \FockSpace_e (\mathrm{B})$,
	$\Supp(o_j) \subseteq \FockSpace_o (\mathrm{A})$, and
	$\Supp(o_j') \subseteq \FockSpace_o (\mathrm{B})$.  The
	probability of the outcome $s$ corresponding to the effect
	$S$, given that the Fermionic system is prepared in state
	$\tau \in \St(\mathrm{AB})$, is provided by the Born rule
	\begin{equation*}
		\Pr(s|\tau)=\Tr[\tau S] = \Tr[P_E \tau P_E S_E + P_O \tau P_O S_O] ,
	\end{equation*}
	which shows us that any separable \ac{POVM} operates on the $E$ and $O$ parts of $\tau$ independently.  
	
	We can now establish that as in quantum theory, in Fermionic theory \ac{SEP} and \ac{LOCC} unambiguous discrimination also achieve the same performances.
	
	\begin{theorem}\label{th:local_discrimination}
		Let $\rho \coloneqq p \ket\psi\bra\psi$ and
		$\sigma \coloneqq q \ket\phi\bra\phi$ be two pure and
		non-normalized states for $p, q \ge 0$ and $p + q = 1$.  The
		optimal \ac{SEP} unambiguous discrimination is
		implementable through \ac{LOCC},
		i.e., $\Psucc^\SEPSet = \Psucc^\LOCCSet$, and its success
		probability reads
		\begin{equation}\label{eq:SEP_success}
			\Psucc^\SEPSet = \Pr(E) \Psucc (\rho_E, \sigma_E) + \Pr(O) \Psucc (\rho_O, \sigma_O) ,
		\end{equation}
		where $\Pr(E) \coloneqq \Tr[(\rho + \sigma) P_E]$, $P_E$ is the projector onto $\sH_E$, $\rho_E \coloneqq P_E \rho P_E / \Pr(E)$, $\sigma_E \coloneqq P_E \sigma P_E / \Pr(E)$, and the same definitions apply for the $O$ sector.
	\end{theorem}
	\begin{proof}
		We now require that the three elements of the \ac{POVM} are separable, i.e., $\Pi_\psi, \Pi_\phi, \Pi_? \in \SEPSet$.
		From Eq.~\eqref{eq:SEP_effect}, a necessary condition for separability is that the operators can be written as $\Pi_\psi = \Pi_\psi^E + \Pi_\psi^O$, 
		$\Pi_\phi = \Pi_\phi^E + \Pi_\phi^O$ and $\Pi_? = \Pi_?^E + \Pi_?^O$, and thus the conditions for unambiguous discrimination of Eq.~\eqref{eq:unambiguous_condition} read
		\begin{align*}
			&\Tr[p\!\ket{\psi_E}\!\bra{\psi_E} \Pi_\phi^E] + \Tr[p\!\ket{\psi_O}\!\bra{\psi_O} \Pi_\phi^O]=0,
			\intertext{and}
			&\Tr[q\!\ket{\phi_E}\!\bra{\phi_E} \Pi_\psi^E] + \Tr[q\!\ket{\phi_O}\!\bra{\phi_O} \Pi_\psi^O]=0.
		\end{align*}
		Since all the operators involved are positive, the
		terms $\Tr[p \ket{\psi_i}\bra{\psi_i} \Pi_\phi^i]$ and
		$\Tr[q \ket{\phi_i}\bra{\phi_i} \Pi_\psi^i]$ for
		$i = E, O$ must be null altogether.  In other words,
		the optimization procedure runs independently on the
		$E$ and $O$ sectors, and the optimal strategy
		corresponds then to first measuring the projectors
		$P_E$, $P_O$ and, depending on the outcome, optimally
		distinguishing between $\rho_E$, $\sigma_E$ or
		$\rho_O$, $\sigma_O$, respectively.  Both steps are
		locally implementable through \ac{SEP}, therefore
		they lead us to the success probability of
		Eq.~\eqref{eq:SEP_success}.
		
		As proved in Ref.~\cite{Ji2005}, in quantum theory
		$\Psucc = \Psucc^\SEPSet = \Psucc^\LOCCSet$.  Namely,
		there exists a \ac{LOCC} quantum protocol for
		distinguishing between $\rho_i$, $\sigma_i$ for
		$i=E,O$, such that its success probability equals the
		optimal one. On the other hand, as shown in
		Ref.~\cite{DAriano2014}, \ac{LOCC} \acp{POVM} on a
		fixed parity sector correspond to \ac{LOCC}
		Fermionic \acp{POVM} in the Jordan-Wigner
		representation, and thus the quantum \ac{LOCC}
		protocol provides a Fermionic \ac{LOCC} protocol.
		Since the optimal unambiguous discrimination between
		states belonging to the same $E$ or $O$ sector is
		\ac{LOCC} implementable, we achieve
		$\Psucc^\SEPSet = \Psucc^\LOCCSet$ in the Fermionic
		case as well.
	\end{proof}
	
	Theorem~\ref{th:local_discrimination} provides us a key result to determine the necessary and sufficient condition for optimal unambiguous discrimination of \ac{SEP} and \ac{LOCC} protocols in the Fermionic theory.
	Since we pointed out at the beginning of the section that the optimal unconstrained Fermionic discrimination protocol is the quantum one,
	in the following we will compare the unconstrained success probabilities of Eqs.~\eqref{eq:psucc_ternary} or \eqref{eq:psucc_binary} to that of \ac{SEP} protocols given by Eq.~\eqref{eq:SEP_success}.
	Under the particular hypotheses, the \ac{SEP} optimal strategy achieves the same performance as the unconstrained one, which indeed proves optimality.
	Furthermore, Lemma~\ref{th:local_discrimination} tells us that if a separable protocol is optimal, then there always exists a \ac{LOCC} one that achieves the same success probability, thus inextricably linking the performances of the two classes.
	It is not restrictive to focus only on the \ac{SEP} discrimination strategies, which is a real advantage since their mathematical definition is much clearer than that of \ac{LOCC}~\cite{Chitambar2014}.
	
	The next lemma introduces the most significant difference with
	respect to quantum theory, proving a necessary condition for a
	pair of Fermionic states to be optimally discriminable through
	\ac{LOCC} \acp{POVM}.
	
	\begin{lemma}[Necessary condition]\label{th:necessary_condition}
		Let $\rho \coloneqq p \ket\psi\bra\psi$ and $\sigma \coloneqq q \ket\phi\bra\phi$ be two pure and non-normalized states, with $p, q \ge 0$ and $p + q = 1$.
		The discrimination protocol through \ac{SEP} is optimal only if
		\begin{equation}\label{eq:argument_condition}
			\arg \braket{\psi_E | \phi_E} = \arg \braket{\psi_O | \phi_O} ,
		\end{equation}
		or if any scalar product $\braket{\psi_E | \phi_E}, \braket{\psi_O | \phi_O}$ is null.
	\end{lemma}
	
	\begin{proof}
		Both success probabilities in Eqs.~\eqref{eq:psucc_ternary}
		and \eqref{eq:psucc_binary} are functions of
		$\abs{\braket{\psi | \phi}} = \sqrt{\abs{\braket{\psi | \phi}}^2}$ or
		\begin{multline*}
			\abs{\braket{\psi | \phi}}^2 = \abs{\braket{\psi_E | \phi_E}}^2 + \abs{\braket{\psi_O | \phi_O}}^2 \\
			+ 2 \abs{\braket{\psi_E | \phi_E}} \cdot \abs{\braket{\psi_O | \phi_O}} \cos\Delta ,
		\end{multline*}
		where
		$\Delta = \arg\braket{\psi_E | \phi_E} -
		\arg\braket{\psi_O | \phi_O}$.
		From the expressions in Eqs.~\eqref{eq:psucc_ternary}
		and~\eqref{eq:psucc_binary}, it is clear that, for
		fixed $\ket{\psi_i},\ket{\phi_i}$ such that
		$\braket{\psi_i | \phi_i}\neq 0$, with $X=E,O$, the
		unconstrained success probability is a function
		$\Psucc(\Delta)$, whose minimum is achieved for
		$\Delta \in 2 \pi \mathbb{Z}$. On the other hand,
		$\Psucc^\SEPSet$ in Eq.~\eqref{eq:SEP_success} is
		independent of the phase shift $\Delta$.  Since
		$\Psucc^\SEPSet \le \Psucc(\Delta)$ for any value of
		$\Delta$, the \ac{SEP} discrimination for
		$\braket{\psi_i | \phi_i}\neq 0$, with $X=E,O$, can
		achieve the performances of the optimal one only if
		$\Delta\in2\pi\mathbb Z$.  The case where
		$\braket{\psi_E | \phi_E} = 0$ or
		$\braket{\psi_O | \phi_O} = 0$ is straightforward
		since the component of $\abs{\braket{\psi | \phi}}$
		depending on $\Delta$ vanishes and one has
		$\Psucc^\SEPSet=\Psucc(\Delta)$. In conclusion, the
		\ac{SEP} protocol achieves optimal performances
		only if the unconstrained one cannot take advantage
		from the relative phase of the $E$ and $O$ parts.
	\end{proof}
	
	\subsection*{Necessary and sufficient conditions for optimal \acs{LOCC} discrimination}
	Based on Lemma \ref{th:necessary_condition}, we now assume the
	complex arguments of the two scalar products being equal, as
	in Eq.~\eqref{eq:argument_condition}, and proceed to derive
	the necessary and sufficient conditions for optimal
	discrimination through \ac{LOCC} \acp{POVM}. As in the
	quantum case, the optimal unconstrained discrimination of the
	states $\{\rho, \sigma\}$ can be ternary or binary.  Moreover,
	in the Fermionic case, one has a broader range of cases since
	the ternary and binary strategies could be applied to
	distinguish the states $\{\rho_E, \sigma_E\}$ and
	$\{\rho_O, \sigma_O\}$ in the even and odd sector,
	respectively.  In the following, we consider all possible
	cases.
	
	In order to use Eqs.~\eqref{eq:psucc_ternary} and
	\eqref{eq:psucc_binary} for calculating
	$\Psucc (\rho_E, \sigma_E)$ and $\Psucc (\rho_O, \sigma_O)$, it
	is convenient to introduce the conditional probability
	distributions $\{p_i, q_i\}$, and the normalized states
	$\ket{\tilde{\psi}_i} \coloneqq \ket{\psi_i} / \norm{\psi_i}$,
	$\ket{\tilde{\phi}_i} \coloneqq \ket{\phi_i} / \norm{\phi_i}$,
	with $i=E,O$, such that
	\begin{equation}\label{eq:notation}
		\begin{aligned}
			&\rho_i = p_i \ket{\tilde{\psi}_i}\bra{\tilde{\psi}_i},
			&&p_i \coloneqq \Tr[\rho_i] = \frac{p \norm{\psi_i}^2}{p \norm{\psi_i}^2 + q \norm{\phi_i}^2} , \\
			&\sigma_i = q_i \ket{\tilde{\phi}_i}
			\bra{\tilde{\phi}_i},
			&&q_i \coloneqq \Tr[\sigma_i] = \frac{q \norm{\phi_i}^2}{p
				\norm{\psi_i}^2 + q \norm{\phi_i}^2}\\
			&i=E,O. 
		\end{aligned}
	\end{equation}
	Given the probability distribution $\{p_i, q_i\}$, the condition for
	the optimal discrimination strategy between $\rho_i$ and $\sigma_i$
	being binary rather than ternary becomes $\X{i} \le \Xi(p_i, q_i)$.
	One can prove with ease that if both $p_E \ge q_E$ and
	$p_O \ge q_O$, then we have $p \ge q$.  The same relation applies for
	the reverse and strict ordering.
	
	Hereafter, we assume that at least one scalar product between
	$\braket{\psi_E | \phi_E}, \braket{\psi_O | \phi_O}$ is
	non-null. Otherwise, we have
	$\braket{\psi | \phi} = \braket{\psi_E | \phi_E} +
	\braket{\psi_O | \phi_O} = 0$
	and the unambiguous discrimination problem reduces to a perfect
	discrimination one, which was solved in
	Ref.~\cite{Lugli2020}.  Moreover, if both states
	$\ket\psi, \ket\phi$ belong to the same $E$ or $O$ sector,
	e.g., $\ket\psi = \ket{\psi_E}$ and $\ket\phi = \ket{\phi_E}$,
	one can straightforwardly apply the corresponding results in
	quantum theory~\cite{Chen2001,Chen2002,Ji2005}.  Indeed, since
	in that case one has
	$\abs{\braket{\psi | \phi}} = \abs{\braket{\tilde{\psi}_E |
			\tilde{\phi}_E}}$,
	$\Pr(O) = 0$ and $\Pr(E) = 1$, Eq.~\eqref{eq:SEP_success}
	leads to $\Psucc = \Psucc^\SEPSet$.
	
	\subsubsection{Ternary case}
	We begin with the unconstrained discrimination being ternary.
	
	\begin{theorem}[Ternary case]\label{th:ternary}
		Let $\rho \coloneqq p \ket\psi\bra\psi$ and
		$\sigma \coloneqq q \ket\phi\bra\phi$ be two pure and
		non-normalized states, with $p, q \ge 0$ and $p + q = 1$.
		If $\ket\psi$ and $\ket\phi$ satisfy
		\begin{equation*}
			\abs{\braket{\psi | \phi}} \le \Xi(p,q),
		\end{equation*}
		then the discrimination protocol through
		\ac{SEP} is optimal if and only if
		$\arg \braket{\psi_E | \phi_E} = \arg
		\braket{\psi_O | \phi_O}$ and
		\begin{equation}\label{eq:ternary_SEP_condition}
			\X{i} \le \Xi(p_i,q_i)
		\end{equation}
		for both $i=E,O$.  See Eqs.~\eqref{eq:Xi}
		and~\eqref{eq:notation} for the definition of $\Xi$
		and
		$\ket{\tilde{\psi}_i},\ket{\tilde{\phi}_i}$,$p_i,q_i$,
		$i=E,O$.
	\end{theorem}
	
	\begin{proof}
		($\Rightarrow$) In Lemma~\ref{th:necessary_condition}, we
		have already proved that Eq.~\eqref{eq:argument_condition}
		is a necessary condition for optimal \ac{SEP}
		discrimination.  Hence, let us consider the case where
		Eq.~\eqref{eq:ternary_SEP_condition} is not satisfied,
		i.e., the optimal discrimination strategy is binary in at
		least one of the $E$ or $O$ sectors.  Without loss of
		generality, we assume that
		$\abs{\braket{\psi_E | \phi_E}} > \Xi(p_E, q_E)$, and
		compare the unconstrained success probability of
		Eq.~\eqref{eq:psucc_ternary} with
		\begin{multline*}
			\Psucc^\SEPSet = \Pr(E) \max\{p_E, q_E\} \left(1 - \abs{\braket{\tilde{\psi}_E | \tilde{\phi}_E}}^2\right) \\
			+ \Pr(O) \left(1 - 2\sqrt{p_O q_O} \abs{\braket{\tilde{\psi}_O | \tilde{\phi}_O}}\right) .
		\end{multline*}
		The above relation arises from
		Lemma~\ref{th:local_discrimination} where we
		substituted Eq.~\eqref{eq:psucc_binary} for
		$\Psucc (\rho_E, \sigma_E)$ and
		Eq.~\eqref{eq:psucc_ternary} for
		$\Psucc (\rho_O, \sigma_O)$.  Since from
		Eq.~\eqref{eq:binary_versus_ternary} the binary case
		is strictly less performing than the ternary, one has
		\begin{equation*}\begin{split}
				\Psucc^\SEPSet &< \Pr(E) (1 - 2\sqrt{p_E q_E} \abs{\braket{\tilde{\psi}_E | \tilde{\phi}_E}}) \\
				&\qquad\qquad+ \Pr(O) (1 - 2\sqrt{p_O q_O} \abs{\braket{\tilde{\psi}_O | \tilde{\phi}_O}}) \\
				&= 1 - 2 \sqrt{pq} \left(\abs{\braket{\psi_E | \phi_E}} + \abs{\braket{\psi_O | \phi_O}}\right) \\
				&\le \Psucc(\rho, \sigma) ,
		\end{split}\end{equation*}
		with the latter inequality being due to the triangle
		inequality. The above relation applies as well when the
		binary discrimination occurs on the $O$ sector or on
		both.
		
		($\Leftarrow$) Suppose now that the states fulfill
		Eqs.~\eqref{eq:argument_condition}, and
		\eqref{eq:ternary_SEP_condition}. Using again
		Lemma~\ref{th:local_discrimination} where we substituted
		Eq.~\eqref{eq:psucc_ternary} for both
		$\Psucc (\rho_E, \sigma_E)$ and
		$\Psucc (\rho_O, \sigma_O)$, the optimal success
		probability for a separable discrimination protocol
		reads
		\begin{equation*}\begin{split}
				\Psucc^\SEPSet &= 1 - 2 \sqrt{pq} \left(\abs{\braket{\psi_E | \phi_E}} + \abs{\braket{\psi_O | \phi_O}}\right) \\
				&= 1 - 2 \sqrt{pq} \abs{\braket{\psi | \phi}} = \Psucc .
		\end{split}\end{equation*}
		where the last equality one is due to
		\eqref{eq:ternary_SEP_condition} (the triangle
		inequality achieves equality if and only if
		Eq.~\eqref{eq:argument_condition} is satisfied).
	\end{proof}
	
	\subsubsection{Binary case}
	
	We are now left with the case where the optimal unconstrained
	strategy is binary.  We point out that if
	$\abs{\braket{\psi | \phi}} > \Xi(p, q)$, then $p \neq q$ and
	analogously, if
	$\abs{\braket{\tilde\psi_i | \tilde\phi_i}} > \Xi(p_i, q_i)$
	one has $p_i \neq q_i$. In the binary scenario,
	$\abs{\braket{\psi | \phi}} > \Xi(p, q)$ and we take
	$p \neq q$ hereafter.
	
	Before providing the necessary and sufficient conditions for
	optimal discrimination through \ac{LOCC} \acp{POVM},
	we prove the following lemma.
	
	\begin{lemma}\label{th:ternary_lemma}
		Let $\rho \coloneqq p \ket\psi\bra\psi$ and
		$\sigma \coloneqq q \ket\phi\bra\phi$ be two pure and
		non-normalized states, with $p, q \ge 0$ and $p + q = 1$.
		If $\ket\psi$ and $\ket\phi$ satisfy
		$\X{i} \le \Xi(p_i, q_i)$ for $i = E$ and $O$, then
		$\abs{\braket{\psi | \phi}} \le \Xi(p, q)$.
	\end{lemma}
	
	\begin{proof}
		Thanks to the triangle inequality applied to Eq.~\eqref{eq:scalar_product}, we know that
		\begin{align*}
			\abs{\braket{\psi | \phi}} &\le \norm{\psi_E} \norm{\phi_E} \X{E} + \norm{\psi_O} \norm{\phi_O} \X{O} ,
			\intertext{and, due to our hypothesis,}
			\abs{\braket{\psi | \phi}} &\le \norm{\psi_E} \norm{\phi_E} \Xi(p_E, q_E) + \norm{\psi_O} \norm{\phi_O} \Xi(p_O, q_O).
		\end{align*}
		Furthermore, we point out that the quantities
		$\Xi(p_i, q_i)$ for $i = E$, $O$ do satisfy
		\begin{align*}
			\norm{\psi_i} \norm{\phi_i} \cdot \Xi(p_i, q_i) &= \norm{\psi_i} \norm{\phi_i} \cdot \Xi(p \norm{\psi_i}^2, q \norm{\phi_i}^2) \\
			&= \left\{\begin{array}{lr}
				\frac{q \norm{\phi_i}^2}{\sqrt{qp}}, & \text{for}\ p \norm{\psi_i}^2 \ge q \norm{\phi_i}^2 \\
				\frac{p \norm{\psi_i}^2}{\sqrt{pq}} & \text{for}\ p \norm{\psi_i}^2 < q \norm{\phi_i}^2 
			\end{array}\right. \\
			& = \frac{\min\{p \norm{\psi_i}^2, q \norm{\phi_i}^2\}}{\sqrt{p q}} ,
		\end{align*}
		so the sum of the above expressions leads us to
		\begin{equation*}\begin{split}
				\abs{\braket{\psi | \phi}} &\le \frac{1}{\sqrt{p q}} \sum_{i = E, O} \min\{p \norm{\psi_i}^2, q \norm{\phi_i}^2\} \\
				&\le \frac{\min\{p, q\}}{\sqrt{p q}} = \Xi(p, q)
		\end{split}\end{equation*}
		and the thesis follows.
	\end{proof}
	
	The previous result tells us that when $\abs{\braket{\psi | \phi}} > \Xi(p, q)$, 
	at least one of the two discrimination protocols for $\{\rho_E,\sigma_E\}$ or 
	$\{\rho_O,\sigma_O\}$ has to be binary.  In particular, either (A) they are both binary or (B) one is binary and the other is ternary.
	We analyze case (A) in Theorem~\ref{th:binary_with_ternary} and
	Proposition~\ref{th:binary_with_null}, whereas case (B) is
	discussed in Theorem~\ref{th:binary_with_binary}.
	
	\begin{theorem}[Binary case A]\label{th:binary_with_ternary}
		Let $\rho \coloneqq p \ket\psi\bra\psi$ and
		$\sigma \coloneqq q \ket\phi\bra\phi$ be two pure and
		non-normalized states with $p, q \ge 0$, $p + q = 1$,
		fulfilling 
		\begin{equation*}
			\abs{\braket{\psi | \phi}} > \Xi(p,q). 
		\end{equation*}
		If we further assume that $p_E, q_E, p_O, q_O > 0$ and
		\begin{equation}
			\X{i} \le \Xi(p_i,q_i) \label{eq:ternary_X_condition},
		\end{equation}
		for either $i = E$ or $O$, then \ac{SEP} discrimination is optimal if and only if
		$\arg \braket{\psi_E | \phi_E} = \arg \braket{\psi_O
			| \phi_O}$, and
		\begin{equation}
			\abs{\braket{\psi | \phi} } -\Xi(p, q) = \Lambda_{\bar{i}}, \label{eq:B_BT_condition}
		\end{equation}
		where $\bar{i} = O, E$ for $i = E, O$, respectively.
		The quantity $\Xi$ and $\ket{\tilde{\psi}_i}$, $\ket{\tilde{\phi}_i}$, $p_i$, $q_i$, $i = E, O$, are defined in Eqs.~\eqref{eq:Xi} and~\eqref{eq:notation},
		respectively, while
		\begin{equation*}
			\Lambda_i \coloneqq \sqrt{\frac{\max\{p\|\psi_i\|^2, q\|\phi_i\|^2\}}{\max\{p, q\}}}(\X{i}-\Xi(p_i,q_i)) .
		\end{equation*}
		
	\end{theorem}
	
	\begin{proof}
		Without loss of generality, we assume that
		Eq.~\eqref{eq:ternary_X_condition} holds for $i = E$.
		Therefore, due to Lemma~\ref{th:ternary_lemma}, one has
		$\X{O} > \Xi(p_O,q_O)$.  We now compare the unconstrained
		success probability $\Psucc$ of Eq.~\eqref{eq:psucc_binary}
		to
		\begin{multline}\nonumber
			\Psucc^\SEPSet = \Pr(E) (1 - 2\sqrt{p_E q_E}\ \X{E}) \\
			+ \max\{p \norm{\psi_O}^2, q \norm{\phi_O}^2\}
			\left(1 - \X{O}^2\right) ,
		\end{multline}
		of Eq.~\eqref{eq:SEP_success} in
		Lemma~\ref{th:local_discrimination} where we replaced
		the expressions of Eqs.~\eqref{eq:psucc_ternary}
		and~\eqref{eq:psucc_binary} for
		$\Psucc(\rho_E, \sigma_E)$ and
		$\Psucc(\rho_O, \sigma_O)$, respectively. If the two
		scalar products $\braket{\psi_E | \phi_E}$ and
		$\braket{\psi_O | \phi_O}$ have the same complex
		argument, the condition for optimality
		$\Psucc = \Psucc^\SEPSet$ can be written in the form
		\begin{align}
			ax^2+by^2+cx+dy+fxy+g=0,
			\label{eq:secdeg}
		\end{align}
		where $x=|\braket{\tilde\psi_E|\tilde\phi_E}|$, $y=|\braket{\tilde\psi_O|\tilde\phi_O}|$, and
		\begin{equation}\label{eq:BTB_params}\begin{split}
				a &= - \max\{p, q\} \norm{\psi_E}^2 \norm{\phi_E}^2 \\
				b &= \max\{p \norm{\psi_O}^2, q \norm{\phi_O}^2\} - \max\{p, q\} \norm{\psi_O}^2 \norm{\phi_O}^2 \\
				c &= 2 \sqrt{p q} \norm{\psi_E} \norm{\phi_E} \\
				d &= 0 \\
				f &= - 2 \max\{p, q\} \norm{\psi_E} \norm{\phi_E} \norm{\psi_O} \norm{\phi_O} \\
				g &= \max\{p, q\} - \max\{p \norm{\psi_O}^2, q
				\norm{\phi_O}^2\} - \Pr(E) .
		\end{split}\end{equation}
		Since $a\neq 0$ by hypothesis, we can solve
		Eq.~\eqref{eq:secdeg} as a second degree equation in
		$x$, as explicitly done in
		Appendix~\ref{app:general_condition}. Upon dividing the
		solution, whose explicit expression can be found in
		Eq.~\eqref{eq:general_solution_E}, by
		$- \max\{p, q\} \norm{\psi_E} \norm{\phi_E}$, one
		obtains the thesis.
		
		On the other hand, if we consider the case where solving
		Eq.~\eqref{eq:ternary_X_condition} is fulfilled for
		$i = O$, Eq.~\eqref{eq:secdeg} in the variable $y$ (see
		Appendix~\ref{app:general_condition}) with the
		appropriate substitution of parameters, then we obtain the
		solution in Eq.~\eqref{eq:general_solution_O}, instead.
		For a full derivation of the quantity $\Lambda_{{i}}$,
		see Appendix~\ref{app:Lambda_X}.
	\end{proof}
	
	In Theorem~\ref{th:binary_with_ternary}, we assume all four
	vectors $\ket{\psi_E}$, $\ket{\psi_O}$, $\ket{\phi_E}$, and
	$\ket{\phi_O}$ are non-normalized, but with the norm strictly
	greater than zero.  Since we are interested in discriminating
	between two states $\rho$, $\sigma$, we must require that at
	least one probability among $p_E$, $p_O$ and $q_E$, $q_O$ is
	non-null.  If both states belong to the same sector,
	i.e., either $p_E = q_E = 0$ or $p_O = q_O = 0$, then the protocol
	reduces to the quantum one, and therefore it is optimally
	\ac{LOCC} implementable.  On the contrary, if they belong
	to different sectors, they are orthogonal and thus perfectly
	distinguishable even with \ac{LOCC}; see
	Ref.~\cite{Lugli2020}. Hereafter, we discuss the cases
	that are not included in Theorem~\ref{th:binary_with_ternary},
	namely, where only one of the probabilities $p_E$, $p_O$,
	$q_E$, $q_O$ is equal to zero.
	
	\begin{proposition}[Binary case A]\label{th:binary_with_null}
		Let $\rho \coloneqq p \ket\psi\bra\psi$ and $\sigma \coloneqq q \ket\phi\bra\phi$ be two pure and non-normalized states for $p, q \ge 0$ and $p + q = 1$ fulfilling
		\begin{equation*}
			\abs{\braket{\psi | \phi}} > \Xi(p,q).
		\end{equation*}
		If either $p_E$, $q_E$, $p_O$, or $q_O$ is null, then
		\ac{SEP} discrimination is optimal if and only if
		any of the following conditions applies:
		\begin{enumerate}
			\item
			$p_i\geq q_i$ for both $i=E,O$, 
			\item
			$q_i\geq p_i$ for both $i=E,O$,
			\item
			$q_i=0$, $p_{\,\bar i}<q_{\,\bar i}$ and $\X{\,\bar i}=\min\{1,p/q\}$,
			\item
			$p_i=0$, $q_{\,\bar i}<p_{\,\bar i}$ and $\X{\,\bar i}=\min\{1,q/p\}$,
		\end{enumerate}
		with $\ket{\tilde{\psi}_i},\ket{\tilde{\phi}_i}$,$p_i,q_i$, $i=E,O$,
		defined as in Eq.~\eqref{eq:notation}.
	\end{proposition}
	
	\begin{proof}
		Without loss of generality, we consider the case where
		$q_E = 0$, i.e., $\norm{\phi_E} = 0$ and $\norm{\phi_O} = 1$,
		namely, in Eq.~\eqref{eq:general_condition} $a$, $c$, $d$, $f$, and
		$\abs{\braket{\psi_E | \phi_E}}$ are null.  Thus, we are
		left with the following condition for optimal \ac{SEP}
		discrimination:
		\begin{equation*}
			b \X{O}^2 + g = 0.
		\end{equation*}
		Now, if $p_O\geq q_O$, we necessarily have $p\geq q$; then, $b=g=0$ and thus 
		\ac{SEP} discrimination is always optimal. 
		On the other hand, if $p_O<q_O$, i.e., $p\|\psi_O\|^2<q$, we have
		\begin{align}
			b &= \begin{cases}
				q-p\|\psi_O\|^2 & p \geq q \\
				q \|\psi_E\|^2 & p < q \label{eq:b_cases}
			\end{cases} \\
			g &= \begin{cases}
				- \left(q - p \norm{\psi_O}^2\right) & p \geq q \\
				-p\|\psi_E\|^2& p< q .
			\end{cases} \nonumber
		\end{align}
		The condition for optimality in the hypothesis of $p_O
		< q_O$ is then
		\begin{align*}
			\X{O} = \min\{1,p/q\}.
		\end{align*}		
		We can analogously evaluate
		Eq.~\eqref{eq:general_condition} for the remaining
		cases and the thesis follows.
	\end{proof}
	
	Finally, we deal with the necessary and sufficient
	condition for optimal \ac{SEP} discrimination if the
	protocol is binary for both the $E$ and $O$ sectors.
	
	\begin{theorem}[Binary case B]\label{th:binary_with_binary}
		Let $\rho \coloneqq p \ket\psi\bra\psi$ and
		$\sigma \coloneqq q \ket\phi\bra\phi$ be two pure and
		non-normalized states for $p, q \ge 0$, $p + q = 1$
		fulfilling
		\begin{equation*}
			\abs{\braket{\psi | \phi}} > \Xi(p, q)\nonumber
		\end{equation*}
		and
		\begin{equation}
			\X{i} > \Xi(p_i,q_i)\label{eq:binary_X_condition}
		\end{equation}
		for both $i = E$, $O$.  Then, \ac{SEP} discrimination is optimal if
		and only if
		\begin{equation}\label{eq:BBB_condition_pq}
			|\norm{\psi_E} \norm{\phi_O} \X{E} - \norm{\psi_O} \norm{\phi_E} \X{O}|=\! \Gamma_E^- + \Gamma_O^- 
		\end{equation}
		for
		$p > q$, otherwise if and only if
		\begin{equation}
			\norm{\psi_O} \norm{\phi_E} \X{E} = \norm{\psi_E} \norm{\phi_E}
			\X{O} \!\pm\!\! (\Gamma_E^+ + \Gamma_O^+ )
			, \label{eq:BBB_condition_qp}
		\end{equation}
		where
		\begin{equation*}
			\Gamma_i^\pm \coloneqq \sqrt{\tfrac{\Pr(i)}{\max\{p, q\}}} \sqrt{(p_i - q_i)^\pm \left(1 - \X{X}^2\right)} ,
		\end{equation*}
		and $(\ \cdot\ )^\pm$ denote the positive and negative parts. The
		quantity $\Xi$ and
		$\ket{\tilde{\psi}_i},\ket{\tilde{\phi}_i}$,$p_i,q_i$, $i=E,O$ are
		defined in Eqs.~\eqref{eq:Xi} and~\eqref{eq:notation}, respectively.
	\end{theorem}
	
	\begin{remark}\label{rm:Gamma}
		We observe that for $p > q$, one has $\Gamma_E^-\Gamma_O^-=0$, and for $p < q$, $\Gamma_E^+\Gamma_O^+=0$.
		The reason is that if we assume, e.g., $p > q$, we have
		either $p_E > q_E$, $p_O > q_O$; $p_E < q_E$,
		$p_O > q_O$ or $p_E > q_E$, $p_O < q_O$.  Hence, at
		least one of the factors $(p_i - q_i)^-$ for
		$i = E, O$ is null. The same argument applies for the
		case $p < q$.
	\end{remark}
	
	\begin{proof} (Theorem~\ref{th:binary_with_binary}) Thanks to
		Lemma~\ref{th:local_discrimination}, we know that
		\begin{multline}
			\Psucc^\SEPSet = \max\{p \norm{\psi_E}^2, q \norm{\phi_E}^2\} \left(1 - \X{E}^2\right) \\
			+ \max\{p \norm{\psi_O}^2, q \norm{\phi_O}^2\} \left(1 - \X{O}^2\right) ,
		\end{multline}
		which has to be compared to
		Eq.~\eqref{eq:psucc_binary} for the unconstrained
		case.  The condition of optimality
		$\Psucc = \Psucc^\SEPSet$ may be rewritten in the form
		of Eq.~\eqref{eq:general_condition}, where $c$, $d$
		are null and
		\begin{align*}
			a &= \max\{p \norm{\psi_E}^2, q \norm{\phi_E}^2\} - \max\{p, q\} \norm{\psi_E}^2 \norm{\phi_E}^2 \\
			b &= \max\{p \norm{\psi_O}^2, q \norm{\phi_O}^2\} - \max\{p, q\} \norm{\psi_O}^2 \norm{\phi_O}^2 \\
			f &= - 2 \max\{p, q\} \norm{\psi_E} \norm{\phi_E} \norm{\psi_O} \norm{\phi_O} \\
			g &= \max\{p, q\} - \max\{p \norm{\psi_E}^2, q \norm{\phi_E}^2\} \\
			& \qquad - \max\{p \norm{\psi_O}^2, q \norm{\phi_O}^2\} .
		\end{align*}
		The solutions satisfy either Eq.~\eqref{eq:general_solution_E} or \eqref{eq:general_solution_O}, namely,
		\begin{align}
			\begin{split}\label{eq:solution_E}
				2 a \X{E} =& - f \X{O} \\
				&\quad\pm \sqrt{(f^2 - 4 a b) \X{O}^2 - 4 a g}
			\end{split} , \\
			\intertext{or}
			\begin{split}\label{eq:solution_O}
				2 b \X{O} =& - f \X{E} \\
				&\quad\pm \sqrt{(f^2 - 4 a b) \X{E}^2 - 4 b g} .
			\end{split}
		\end{align}
		We now have 
		\begin{align*}
			a &= \begin{cases}
				p \norm{\psi_E}^2 \norm{\phi_O}^2 & p_E > q_E \\
				\norm{\phi_E}^2 \left(q - p \norm{\psi_E}^2\right) & p_E < q_E
			\end{cases} \\
			b &= \begin{cases}
				p \norm{\psi_O}^2 \norm{\phi_E}^2 & p_O > q_O \\
				\norm{\phi_O}^2 \left(q - p \norm{\psi_O}^2\right) & p_O< q_O
			\end{cases} \\
			g &= \begin{cases}
				0 & p_E > q_E,\ p_O > q_O \\
				p \norm{\psi_O}^2 - q \norm{\phi_O}^2 & p_E > q_E,\ p_O < q_O \\
				p \norm{\psi_E}^2 - q \norm{\phi_E}^2 & p_E < q_E,\ p_O > q_O . \\
			\end{cases}
		\end{align*}
		Moreover, we estimate the following quantity in the
		radicands of Eqs.~\eqref{eq:solution_E}
		and~\eqref{eq:solution_O}:
		\begin{equation*}
			f^2 - 4 a b = \begin{cases}
				0 & p_E > q_E,\ p_O > q_O \\
				4 a g & p_E > q_E,\ p_O < q_O \\
				4 b g & p_E < q_E,\ p_O > q_O . \\
			\end{cases}
		\end{equation*}
		
		We now consider the three cases compatible with the requirements discussed in the text, viz., that $p_E \neq q_E$, $p_O \neq q_O$ due to the hypotheses in Eq.~\eqref{eq:binary_X_condition} and that $p_E > q_E$ or $p_O > q_O$, due to $p > q$.
		Thus, we first point out that for $p_E > q_E$, $p_O > q_O$, both
		Eqs.~\eqref{eq:solution_E} and~\eqref{eq:solution_O}
		reduce to
		\begin{equation*}
			\norm{\psi_E} \norm{\phi_O} \X{E} = \norm{\psi_O} \norm{\phi_E} \X{O} .
		\end{equation*}
		Indeed, in such a case, we have that $\Gamma_E^-,
		\Gamma_O^- = 0$
		and the thesis follows.
		Second, for $p_E > q_E$, $p_O < q_O$, we consider the expression of Eq.~\eqref{eq:solution_E} and rewrite it as
		\begin{equation*}
			a \X{E} = - \frac{f}{2} \X{O} \pm \sqrt{a g \left(\X{O}^2 - 1\right)}
		\end{equation*}
		or, equivalently,
		\begin{multline}\label{eq:solution_E_bis}
			\norm{\psi_E} \norm{\phi_O} \X{E} = \norm{\psi_O} \norm{\phi_E} \X{O} \\
			\pm \frac{1}{\sqrt{p}} \sqrt{\left(q \norm{\phi_O}^2 - p \norm{\psi_O}
				^2\right) \left(1 - \X{O}^2\right)} .
		\end{multline}
		Due to our hypotheses, we have that $\Gamma_E^- = 0$ and, thus, 
		Eq.~\eqref{eq:solution_E_bis} reduces to the thesis.
		In the case where $p_E < q_E$, $p_O > q_O$, we rewrite 
		Eq.~\eqref{eq:solution_O} as
		\begin{equation*}
			b \X{O} = - \frac{f}{2} \X{E} \pm \sqrt{b g \left(\X{E}^2 - 1\right)} ,
		\end{equation*}
		such that
		\begin{multline}\label{eq:solution_O_bis}
			\norm{\psi_O} \norm{\phi_E} \X{O} = \norm{\psi_E} \norm{\phi_O} \X{E} \\
			\pm \frac{1}{\sqrt{p}} \sqrt{\left(q \norm{\phi_E}^2 - p \norm{\psi_E}^2\right)\left(1 - \X{E}^2\right)} .
		\end{multline}
		The same argument as in the previous case applies here: 
		$\Gamma_O^- = 0$ since $p_O > q_O$, and Eq.~\eqref{eq:solution_O_bis} 
		becomes the thesis.
		Exactly the same steps can be taken for $p < q$.
		We only point out that one has to consider Eq.~\eqref{eq:solution_E} when 
		$p_E < q_E$, $p_O > q_O$ and Eq.~\eqref{eq:solution_O} for $p_E > q_E$, 
		$p_O < q_O$, instead. The final result is the condition in the statement.
	\end{proof}
	
	\section{Ancilla assisted discrimination}
	In this section, we prove that the limits to the performances of
	unambiguous \ac{LOCC} discrimination in Fermionic theory
	may be overcome by taking advantage of an ancillary system
	shared by Alice and Bob.  Let us take the simplest sharable
	system, i.e., two local Fermionic modes, in the pure and
	normalized state,
	\begin{equation}\label{eq:ancilla}
		\ket\omega \coloneqq a \ket{00} + b \ket{11} , \quad \abs{a},\abs{b} \neq 0 .
	\end{equation}
	We are now interested in distinguishing between the two pure and non-normalized, states $\rho' \coloneqq \rho \otimes \ket\omega\bra\omega$, $\sigma' \coloneqq \sigma \otimes \ket\omega\bra\omega$.
	In the following result, we prove that any pair $\rho,\sigma$ can be optimally 
	discriminated via \ac{LOCC} if and only if the ancillary system is in a 
	\emph{maximally entangled} state.
	
	\begin{theorem}
		We can always optimally and unambiguously discriminate through \ac{SEP} every pair of pure and non-normalized states $\rho \coloneqq p \ket\psi\bra\psi$, $\sigma \coloneqq q \ket\phi\bra\phi$, for $p, q \ge 0$ and $p + q = 1$, by means of an ancillary system in the state $\ket\omega$ as in Eq.~\eqref{eq:ancilla}, if and only if
		\begin{equation*}
			\ket\omega = \frac{1}{\sqrt2} \left(\ket{00} + e^{i\varphi}\ket{11}\right) , \quad \varphi \in [0, 2\pi) .
		\end{equation*}
	\end{theorem}
	
	\begin{proof}
		We can write $\rho' = p\ket{\psi'}\bra{\psi'}$ and $\sigma' = q\ket{\phi'}\bra{\phi'}$ for
		\begin{align*}
			\ket{\psi'} &\coloneqq \ket\psi \otimes \ket\omega = \ket{\psi_E'} + \ket{\psi_O'} , \\
			\ket{\phi'} &\coloneqq \ket\phi \otimes \ket\omega = \ket{\phi_E'} + \ket{\phi_O'}
		\end{align*}
		with $\ket{\psi'_E} = a \ket{\psi_E 00} + b \ket{\psi_O 11}$, $\ket{\psi'_O} = b \ket{\psi_E 11} + a \ket{\psi_O 00}$, $\ket{\phi'_E} = a \ket{\phi_E 00} + b  \ket{\phi_O 11}$, and $\ket{\phi'_O} = b \ket{\phi_E 11} + a \ket{\phi_O 00}$.
		Hence, the scalar products of the $E$ and $O$ parts read
		\begin{align*}
			\braket{\psi_E' | \phi_E'} &= \abs{a}^2 \braket{\psi_E | \phi_E} + \abs{b}^2 \braket{\psi_O | \phi_O} , \\
			\braket{\psi_O' | \phi_O'} &= \abs{a}^2 \braket{\psi_O | \phi_O} + \abs{b}^2 \braket{\psi_E | \phi_E} .
		\end{align*}
		
		($\Rightarrow$) Lemma~\ref{th:necessary_condition} requires 
		Eq.~\eqref{eq:argument_condition} to be satisfied by the states for \ac{SEP} 
		discrimination to be optimal.
		Namely,
		\begin{equation*}
			\arg\braket{\psi_E' | \phi_E'} = \arg\braket{\psi_O' | \phi_O'} ,
		\end{equation*}
		which may be rewritten as
		\begin{equation}\label{eq:phase_condition}
			\left(\abs{a}^4 - \abs{b}^4\right) \left(e^{i \Delta} - e^{- i \Delta}\right) = 0
		\end{equation}
		where $\Delta = \arg\braket{\psi_E | \phi_E} - \arg\braket{\psi_O | \phi_O}$; see Ref.\footnote{We point out that we can assume $\braket{\psi_E' | \phi_E'} \ge 0$ and $\braket{\psi_O' | \phi_O'} = \abs{\braket{\psi_O' | \phi_O'}} e^{i\Delta}$, since we are dealing with an argument difference. Moreover, one could achieve Eq.~\eqref{eq:phase_condition} through \[ \arg(z + w) = \frac{1}{2i} \left[\ln(z + w) - \ln(\bar{z} + \bar{w})\right] . \]}.
		The above expression is satisfied by any $\Delta$ if and only if $\abs{a}^2 = \abs{b}^2 = 1/2$, i.e., for a maximally entangled ancilla.
		
		($\Leftarrow$) We now suppose $\abs{a}^2 = \abs{b}^2 = 1/2$ and prove sufficiency.
		The new states satisfy the following properties:
		\begin{enumerate}
			\item $\braket{\psi' | \phi'} = \braket{\psi | \phi}$, hence $\abs{\braket{\psi' | \phi'}} \le \Xi(p, q)$ if and only if $\abs{\braket{\psi | \phi}} \le \Xi(p, q)$.
			\item\label{itm:amplitudes} $\norm{\psi_E'} = \norm{\psi_O'} = \norm{\phi_E'} = \norm{\phi_O'} = 1/\sqrt{2}$.
			\item ${\Pr}'(E) = \Tr[(\rho' + \sigma') P_E] = 1 / 2 = {\Pr}'(O)$.
			\item\label{itm:probabilities}
			$p_E' = \Tr[\rho' P_E] / {\Pr}'(E) = p =
			p_O'$
			and $q_E' = q_O' = q$.  Thus,
			$\Xi(p, q) = \Xi(p_E', q_E') = \Xi(p_O',
			q_O')$.
			\item\label{itm:scalar_product}
			$\braket{\psi_E' | \phi_E'} = 1/2
			\left(\braket{\psi_E | \phi_E} +
			\braket{\psi_O | \phi_O} \right) =
			\braket{\psi_O' | \phi_O'}$.
		\end{enumerate}
		[see Eq.~\eqref{eq:Xi} for the definition of $\Xi$.]
		Thanks to properties~\ref{itm:amplitudes} and
		\ref{itm:scalar_product}, we have that
		\begin{equation*}
			\begin{aligned}
				\braket{\psi' | \phi'} &= \braket{\psi_E' | \phi_E'} + \braket{\psi_O' | \phi_O'} = 2 \braket{\psi_E' | \phi_E'} \\
				&= 2 \norm{\psi_E'} \norm{\phi_E'}
				\braket{\tilde{\psi}_E' | \tilde{\phi}_E'} =
				\braket{\tilde{\psi}_E' | \tilde{\phi}_E'} ,\\
				\braket{\psi' | \phi'} &= 2 \braket{\psi_O' | \phi_O'} = \braket{\tilde{\psi}_O' | \tilde{\phi}_O'} .
			\end{aligned}
		\end{equation*}
		The above results lead us to the fact that either
		$\abs{\braket{\psi | \phi}}$,
		$\abs{\braket{\psi' | \phi'}}$,
		$\abs{\braket{\tilde{\psi}_E' | \tilde{\phi}_E'}}$ and
		$\abs{\braket{\tilde{\psi}_O' | \tilde{\phi}_O'}}$ are
		all smaller than $\Xi(p, q)$ or they are all greater.
		We observe that if the unconstrained discrimination
		between the two original states $\rho$, $\sigma$ is
		ternary, then Theorem~\ref{th:ternary} ensures us that
		\ac{SEP} discrimination between $\rho'$ and
		$\sigma'$ is optimal.  Otherwise, if
		$\abs{\braket{\psi | \phi}} > \Xi(p, q)$, we have to
		further prove the validity of either
		Eq.~\eqref{eq:BBB_condition_pq} or
		\eqref{eq:BBB_condition_qp}.  However,
		property~\ref{itm:probabilities} tells us that either
		$p_E, p_E', p_O' > q_E, q_E', q_O'$ or vice versa.
		Therefore, the quantities $\Gamma^\pm_{i} = 0$ for both
		$i = E$ and $O$, as we have shown in
		Remark~\ref{rm:Gamma}.  Eventually, both equations
		\begin{align*}
			\norm{\psi_E'} \norm{\phi_O'} \abs{\braket{\tilde{\psi}_E' | \tilde{\phi}_E'}} &= \norm{\psi_O'} \norm{\phi_E'} \abs{\braket{\tilde{\psi}_O' | \tilde{\phi}_O'}}
			\intertext{and}
			\norm{\psi_O'} \norm{\phi_E'} \abs{\braket{\tilde{\psi}_E' | \tilde{\phi}_E'}}  &= \norm{\psi_E'} \norm{\phi_E'} \abs{\braket{\tilde{\psi}_O' | \tilde{\phi}_O'}}
		\end{align*}
		are satisfied by the states $\rho'$ and $\sigma'$, which proves sufficiency.
	\end{proof}
	
	\section{Discussion}
	
	As we have seen so far, the behavior of Fermionic unambiguous discrimination strategies genuinely differs from their quantum counterparts when we focus on their performances under locality restriction.
	Fermionic protocols optimally distinguish two states through \ac{LOCC} only if strict conditions on the preparations are fulfilled.
	However, we remark that the relationship between the \ac{LOCC} and \ac{SEP} classes still remains unchanged in terms of binary discrimination and, as we have seen in Lemma~\ref{th:local_discrimination}, the two classes achieve an identical maximum probability of success in unambiguous discrimination.
	
	We proved that the limits of Fermionic \ac{LOCC} discrimination are completely overcome if we take advantage of an ancillary system.
	A rather striking result, which echoes a similar constraint in the case of minimum-error discrimination~\cite{Lugli2020}, is that the ancilla is required to be maximally entangled in order to attain optimality for every pair of preparations.
	
	The comparison between the current results and those pertaining to the unambiguous discrimination of quantum states through \ac{LOCC} leaves room for some remarks.
	In Ref.~\cite{Ji2005}, the authors prove the existence of a particular basis for Alice that allows Bob to either perfectly or unambiguously distinguish between two local states, thus achieving the optimal unconstrained performance.
	However, in the Fermionic case, Theorems~\ref{th:ternary_lemma}, \ref{th:binary_with_ternary}, \ref{th:binary_with_binary} and Proposition~\ref{th:binary_with_null} show that such a basis does not exist in those cases where the conditions of the theorems are not met.
	Most notably, if we consider a pair of Fermionic states for which the \ac{LOCC} unambiguous discrimination is strictly suboptimal, the strategy of Ref.~\cite{Ji2005} would inevitably involve vectors for Alice or Bob that are forbidden by the parity superselection rule.
	
	In this paper, we obtained the condition for optimal unambiguous discrimination through \ac{LOCC} of Fermionic states.
	Our proof is based on the comparison between the success probability of the constrained and unconstrained strategies.
	We may wonder if another proof could be derived where the conditions on the states were expressed in a more algebraic form, as those developed for the conclusive case of Ref.~\cite{Lugli2020}.
	This is left as an open question for further development.
	
	A natural sequel of the present work would address the conditions for unambiguous discrimination through local resources in the case of \emph{mixed} Fermionic states.
	Nonetheless, such a task remains unresolved even in the quantum realm~\cite{Dusek2000,Chefles2004,Sugimoto2009,Chitambar2014b}, where strict conditions on the preparations have been derived only for particular circumstances.
	From full rigorous development, we would expect the conditions on the density matrices in the Fermionic case to be stronger than in the quantum one.
	Since we have already observed that some further requirements arise to unambiguously distinguish between two pure Fermionic states, we have evidence that mixed Fermionic preparations are subject to more stringent prerequisites.
	In Ref.~\cite{Chefles2004}, \citeauthor{Chefles2004} derives a necessary and sufficient condition for the feasibility of unambiguous discrimination of mixed quantum states through \ac{LOCC}, i.e., for a non-null success probability.
	However, the superselection rule cannot but increase the chances of having states not unambiguously distinguishable to any extent.
	
	\begin{acknowledgments}
		A.~T. acknowledges financial support from the Elvia and Federico Faggin Foundation through the Silicon Valley Community Foundation, Grant No.~2020-214365.
		This work was supported by MIUR Dipartimenti di Eccellenza 2018-2022 project F11I18000680001.
	\end{acknowledgments}
	
	\appendix
	
	\section{General condition for \acs{SEP} discrimination}\label{app:general_condition}
	In Theorems~\ref{th:binary_with_ternary}, \ref{th:binary_with_binary} and in Proposition~\ref{th:binary_with_null}, we derive the necessary and sufficient conditions for \ac{SEP} unambiguous discriminations by taking advantage of Theorem~\ref{th:local_discrimination}.
	Indeed, the latter states that the optimal \ac{SEP} discrimination protocol is \ac{LOCC} implementable with success probability
	\begin{equation*}
		\Psucc^\SEPSet = \Pr(E) \Psucc (\rho_E, \sigma_E) + \Pr(O) \Psucc (\rho_O, \sigma_O) . \tag{\ref{eq:SEP_success}}
	\end{equation*}
	We then compare the success probability of the \ac{SEP} protocol given by Eq.~\eqref{eq:SEP_success} with the unconstrained one of Eq.~\eqref{eq:psucc_binary}.
	Moreover, since we know from Lemma~\ref{th:necessary_condition} that we can optimally distinguish between two states only if
	\begin{equation*}
		\arg \braket{\psi_E | \phi_E} = \arg \braket{\psi_O | \phi_O} ,
	\end{equation*}
	in the proofs of the aforementioned results we take for granted that such a condition is satisfied.
	As a result, we have that the \ac{SEP} discrimination is optimal if and only if
	\begin{multline}\label{eq:general_condition}
		a \X{E}^2 + b \X{O}^2 + c \X{E} + d \X{O} \\
		+ f \X{E} \cdot \X{O} + g = 0 ,
	\end{multline}
	where the parameters $a$, $b$, $c$, $d$, $f$, and $g$ depend on the case under scrutiny.
	In the most general case, Eq.~\eqref{eq:general_condition} has the following solutions:
	\begin{widetext}\begin{align}
			- 2 a \X{E} &= f \X{O} + c \pm \sqrt{(f^2 - 4 a b) \X{O}^2 + (2 c f - 4 a d) \X{O} + c^2 - 4 a g} & \text{for}\ a &\neq 0 \label{eq:general_solution_E} , \\
			- 2 b \X{O} &= f \X{E} + d \pm \sqrt{(f^2 - 4 a b) \X{E}^2 + (2 d f - 4 b c) \X{E} + d^2 - 4 b g} & \text{for}\ b &\neq 0 . \label{eq:general_solution_O}
	\end{align}\end{widetext}
	In Theorems~\ref{th:binary_with_ternary}, \ref{th:binary_with_binary} and in Proposition~\ref{th:binary_with_null}, we take advantage of the additional hypothesis on the states to further simplify the above expressions.
	The choice between Eqs.~\eqref{eq:general_solution_E} and \eqref{eq:general_solution_O} is only based on a simplification and convenience criterion in the analysis of the discrimination cases.
	
	\section{Derivation of quantity $\Lambda_i$}\label{app:Lambda_X}
	In Theorem~\ref{th:binary_with_ternary}, we covered the case where the \ac{SEP} discrimination is binary on either the $E$ or $O$ sector and ternary in the other.
	In particular, in the proof of the theorem, we assumed $\X{E} \le \Xi(p_E, q_E)$ and $\X{O} > \Xi(p_O, q_O)$, so that the necessary and sufficient condition for optimal \ac{SEP} discrimination reads
	\begin{multline*}
		2 a \X{E} = - f \X{O} - c \\
		\pm \sqrt{(f^2 - 4 ab) \X{O}^2 + 2 c f \X{O} + c^2 - 4 a g } .
	\end{multline*}
	The values for parameters $a$, $b$, $c$, $f$, and $g$ are expressed in Eq.~\eqref{eq:BTB_params} so that
	\begin{multline*}
		f^2 - 4 a b = 4 {\max}^2\{p, q\} \norm{\psi_E}^2 \norm{\phi_E}^2 \norm{\psi_O}^2 \norm{\phi_O}^2 \\
		+ 4 \max\{p, q\} \norm{\psi_E}^2 \norm{\phi_E}^2 \left(\max\{p \norm{\psi_O}^2, q \norm{\phi_O}^2\} \right. \\
		\left. - \max\{p, q\} \norm{\psi_O}^2 \norm{\phi_O}^2\right) .
	\end{multline*}
	Since we assumed $p_E, q_E \neq 0$, i.e., $\norm{\psi_E}, \norm{\phi_E} \neq 0$, we may reformulate the above expression as
	\begin{multline*}
		\norm{\psi_E} \norm{\phi_E} \X{E}
		= - \norm{\psi_O} \norm{\phi_O} \X{O} \\
		+ \sqrt{\frac{p q}{{\max}^2\{p, q\}}} \pm \sqrt{\frac{\Theta}{\max\{p, q\}}} ,
	\end{multline*}
	for
	\begin{equation*}
		\Theta = \alpha \X{O}^2 + \beta \X{O} + \gamma ,
	\end{equation*}
	where we define
	\begin{align*}
		\alpha &= \max\{p \norm{\psi_O}^2, q \norm{\phi_O}^2\} \\
		\beta &= - 2 \sqrt{p q} \norm{\psi_O} \norm{\phi_O} \\
		\begin{split}
			\gamma &= \max\{p, q\} - \Pr(E) - \max\{p \norm{\psi_O}^2, q \norm{\phi_O}^2\} \\
			&\qquad + \frac{pq}{\max\{p, q\}} .
		\end{split}
	\end{align*}
	We prove with ease that
	\[ \frac{p q}{\max\{p, q\}} = \min\{p, q\} , \]
	so that we find
	\begin{multline*}
		\norm{\psi_E} \norm{\phi_E} \X{E} = - \norm{\psi_O} \norm{\phi_O} \X{O} + \Xi(p, q) \\
		\pm \sqrt{\frac{\Theta}{\max\{p, q\}}},
	\end{multline*}
	and
	\begin{align*}
		\gamma &= \max\{p, q\} + \min\{p, q\} - \Pr(E) \\
		&\qquad - \max(p \norm{\psi_O}^2, q \norm{\phi_O}^2) \\
		&= \Pr(O) - \max(p \norm{\psi_O}^2, q \norm{\phi_O}^2) \\
		&= \min(p \norm{\psi_O}^2, q \norm{\phi_O}^2) .
	\end{align*}
	Now we can rewrite the above condition as
	\begin{align*}
		\Y E+\Y O-\Xi(p,q)=\pm\sqrt{\frac\Theta{\max\{p,q\}}},
	\end{align*}
	and remembering that by the hypotheses of Theorem~\ref{th:binary_with_ternary},f one has $\Y E+\Y O=|\braket{\phi|\psi}|>\Xi(p,q)$, we obtain
	\begin{align*}
		|\braket{\psi|\phi}|-\Xi(p,q)=\sqrt{\frac\Theta{\max\{p,q\}}}.
	\end{align*}
	The term $\Theta$ eventually reads
	\begin{align*}
		\Theta &= \left(\sqrt{\max(p \norm{\psi_O}^2, q \norm{\phi_E}^2)} \X{O} \right. \\
		&\left. - \sqrt{\min(p \norm{\psi_O}^2, q \norm{\phi_O}^2)}\right)^2 \\
		&=\sqrt{\max\{p\|\psi_O\|^2,q\|\phi_O\|^2\}}(\X O-\Xi(p_O,q_O)),
	\end{align*}
	and we finally achieve Eq.~\eqref{eq:B_BT_condition}.
	
	\bibliography{bibliography}
\end{document}